\documentclass[12pt]{article}
\usepackage{enumerate}
\usepackage{amsfonts}
\usepackage{amsmath}
\usepackage{amssymb}
\usepackage{amsthm}

\def\rank{\mathrm{rank}}
\def\supp{\mathrm{supp}}

\newcommand{\Tr}{\mathrm{Tr}}

\newcounter{defin}  \newcounter{lemma}  \newcounter{theorem}
\newcounter{property} \newcounter{corol}  \newcounter{remark} \newcounter{example}

\newenvironment{lemma}{\par\refstepcounter{lemma}
     \textbf{Lemma \thelemma.} }{\rm\par}
\newenvironment{theorem}{\par\refstepcounter{theorem}
     \textbf{Theorem \thetheorem.}\ }{\rm\par}
\newenvironment{property}{\par\refstepcounter{property}
     \textbf{Proposition \theproperty.}\ }{\rm\par}
\newenvironment{corollary}{\par\refstepcounter{corol}
     \textbf{Corollary \thecorol.} }{\rm\par}
\newenvironment{definition}{\par\refstepcounter{defin}
     \textbf{Definition \thedefin.}\ }{\rm\par}
\newenvironment{remark}{\par\refstepcounter{remark}
     \textbf{Remark \theremark.}}{\rm\par}
\newenvironment{example}{\par\refstepcounter{example}
     \textbf{Example \theexample.}}{\rm\par}

\begin{document}
\title{Conditions for equality between entanglement-assisted and unassisted classical capacities of a quantum channel}
\author{M.E. Shirokov\\
Steklov Mathematical Institute, RAS, Moscow\\
msh@mi.ras.ru}
\date{}
\maketitle

\begin{abstract}
Several relations between the Holevo capacity and the entanglement-assisted classical capacity of a quantum channel are proved, necessary and sufficient conditions for their coincidence are obtained. In particular, it is shown that these capacities coincide if (correspondingly, only if) the channel  (correspondingly, the $\chi$--essential part of the channel) belongs to the class of classical-quantum channels (the $\chi$-essential part is a restriction of a channel obtained by discarding all states useless for transmission of classical information). The obtained conditions and their corollaries are generalized to channels with linear constraints. By using these conditions it is shown that the question of coincidence of
the Holevo capacity and the entanglement-assisted classical capacity depends on the constraint (even for classical-quantum channels).

Properties of the difference between the quantum mutual information and the $\chi$-function (constrained Holevo capacity) of a quantum channel are explored.
\end{abstract}

\section{Introduction}

Informational properties of a quantum channel are characterized by a
number of different capacities defined by type of transmitted
information, by additional resources used to increase the rate of
this transmission, by security requirements, etc.

Central roles in analysis of transmission of  classical
information through a quantum channel $\Phi$ are played by the
Holevo capacity $\bar{C}(\Phi)$, the classical (unassisted) capacity
$C(\Phi)$ and the entanglement-assisted (classical) capacity
$C_{\mathrm{ea}}(\Phi)$ of this channel. The first of them is
defined as the maximal rate of information transmission between
transmitter and receiver (generally called Alice and Bob) when
nonentangled block coding is used by Alice and arbitrary measurement
is used by Bob, the second one  differs form the first by
possibility to use arbitrary block coding by Alice while the
entanglement-assisted capacity is defined as the maximal rate of
information transmission between Alice and Bob under the assumption
that they share a common entangled state, which can be used in block
coding by Alice to increase the rate of information transmission~\cite{BSST,N&Ch}.

By the operational definitions $\bar{C}(\Phi)\le C(\Phi)\le C_{\mathrm{ea}}(\Phi)$. During a long time it was conjectured that $\bar{C}(\Phi)=C(\Phi)$ for any channel $\Phi$ until Hastings showed
existence of a counter-example to the additivity conjecture~\cite{H}. Nevertheless, the equality $\bar{C}(\Phi)=C(\Phi)$ holds
for a large class of channels including the noiseless channel, all
unital qubit channels, all entanglement-breaking channels and many
other concrete examples. In contract to this, possibility of the
strict inequality $C(\Phi)< C_{\mathrm{ea}}(\Phi)$ was initially
obvious, since the superdense coding implies that
$C_{\mathrm{ea}}(\Phi)=2C(\Phi)>0$ if $\Phi$ is the noiseless
channel. But there exist channels, for which
\begin{equation}\label{coincidence}
\bar{C}(\Phi)=C(\Phi)=C_{\mathrm{ea}}(\Phi)>0
\end{equation}
(as an example one can consider the channel $\rho\mapsto\sum_k\langle k|\rho|k\rangle|k\rangle\langle k|$, where
$\{|k\rangle\}$ is an orthonormal basis). Hence the question
"How can the class of channels for which~\eqref{coincidence} holds
be characterized?" naturally arises. In contrast to an intuitive
point of view this class does not coincide with the class of
entanglement-breaking channels: despite the fact that these channels
annihilate entanglement of any state shared by Alice and Bob, their
entanglement-assisted capacity may be greater then the classical
unassisted capacity~\cite{BSST}. On the other hand, in~\cite{BSST+} an example of non-entanglement-breaking channel
for which
$C_{\mathrm{ea}}(\Phi)=\bar{C}(\Phi)$ is described (see Example~\ref{n-e-b-ch} in Section~2.3 below). A step in finding
answer to the above question was recently made in~\cite{H-CQM}, where a criterion of
\eqref{coincidence} for the class of q-c
channels defined by quantum observables is obtained.

In this paper some relations between the capacities $\bar{C}(\Phi)$ and $C_{\mathrm{ea}}(\Phi)$ as well as necessary and sufficient conditions for the equality $\bar{C}(\Phi)=C_{\mathrm{ea}}(\Phi)$ are obtained
(Proposition~\ref{simple}, Theorems~\ref{coin-crit} and~\ref{e-b-ch}). In particular, it is shown that the equality $\bar{C}(\Phi)=C_{\mathrm{ea}}(\Phi)$ holds if (correspondingly, only if) the channel $\Phi$ (correspondingly, the $\chi$--essential part of the channel $\Phi$) belongs to the class of classical-quantum channels (the $\chi$-essential part is defined as a restriction of a
channel to the set of states supported by the minimal subspace containing elements of \emph{all} ensembles optimal for this channel in the sense of the Holevo capacity, see Definition
~\ref{ef-subch}).

Since in dealing with infinite dimensional channels
it is necessary to impose particular constraints on the choice of input
code-states, we also consider conditions for coincidence of the entanglement-assisted capacity
with the Holevo capacity for quantum channels with linear constraints
(Propositions~\ref{coin-crit+}~and~\ref{degradable-cv}). By using these conditions it is shown that even in the case
of classical-quantum channels the question of coincidence of the above capacities depends on the form of the constraint (Example~\ref{classical-costrained},
Proposition~\ref{e-b-ch-c+}).

In Section~4 properties of the difference between the quantum mutual information and the $\chi$-function (the constrained Holevo capacity) of a quantum channel (considered as a function of an input state) are studied (Theorem~\ref{delta-fun-p}). In particular,
the sense of the maximal value of this function as a parameter
characterizing "noise level" of a quantum channel is shown.

\section{Unconstrained channels}

Let $\mathcal{H}_A$, $\mathcal{H}_B$ and $\mathcal{H}_E$ be finite dimensional Hilbert spaces.
In what follows
$\Phi\colon\mathfrak{S}(\mathcal{H}_A)\to\mathfrak{S}(\mathcal{H}_B)$ is a quantum
channel and $\widehat{\Phi}\colon\mathfrak{S}(\mathcal{H}_A)\to\mathfrak{S}(\mathcal{H}_E)$ is its
complementary channel, defined uniquely up to unitary equivalence~\cite{H-comp-ch}.\footnote{The quantum channel $\widehat{\Phi}$ is also
called \emph{conjugate} to the channel $\Phi$ \cite{KMNR}.}\smallskip

Let $H(\rho)$ and $H(\rho\|\sigma)$ be respectively the von Neumann
entropy of the state $\rho$ and the quantum relative entropy of the
states $\rho$ and $\sigma$~\cite{N&Ch}.\smallskip

The Holevo capacity of the channel $\Phi$ can be defined as follows
\begin{equation}\label{chi-cap}
\bar{C}(\Phi)=\max_{\rho\in\mathfrak{S}(\mathcal{H}_A)}\chi^{}_\Phi(\rho),
\end{equation}
where
\begin{equation}\label{chi-fun}
\chi^{}_\Phi(\rho)=\max_{\sum_i\pi_i\rho_i=\rho}\,\sum_i\pi_i
H(\Phi(\rho_i)\|\Phi(\rho))
\end{equation}
is the $\chi$-function  of the
channel $\Phi$~\cite{H-Sh-1}. Note that
\begin{equation}\label{chi-fun+}
\chi^{}_\Phi(\rho)=H(\Phi(\rho))-\hat{H}_\Phi(\rho),
\end{equation}
where $\hat{H}_\Phi(\rho)=\min_{\sum_i\pi_i\rho_i=\rho}\,\sum_i\pi_i
H(\Phi(\rho_i))$ is the convex hull of the function $\rho\mapsto H(\Phi(\rho))$. By concavity of this function the above minimum can be taken over ensembles of pure states. An ensemble $\{\pi_i,\rho_i\}$ of pure states called \emph{optimal for the channel $\Phi$} if (cf.~\cite{Sch-West})
$$
\bar{C}(\Phi)=\chi^{}_\Phi(\bar{\rho})=\sum_i\pi_i H(\Phi(\rho_i)\|\Phi(\bar{\rho})),
\quad\bar{\rho}=\sum_i\pi_i\rho_i.
$$

By the Holevo-Schumacher-Westmoreland theorem the classical  capacity of the channel $\Phi$ can
be expressed by the following regularization formula
$$
C(\Phi)=\lim_{n\to+\infty} n^{-1}\bar{C}(\Phi^{\otimes n}).
$$

By the Bennett-Shor-Smolin-Thapliyal theorem the entanglement-assisted capacity of the
channel $\Phi$ is determined as follows
\begin{equation}\label{ea-cap}
C_{\mathrm{ea}}(\Phi)=\max_{\rho\in\mathfrak{S}(\mathcal{H}_A)}I(\rho, \Phi),
\end{equation}
where $I(\rho, \Phi)=H(\rho)+H(\Phi(\rho))-H(\widehat{\Phi}(\rho))$ is the quantum mutual information of the channel $\Phi$ at the state
$\rho$~\cite{N&Ch}.\medskip

By the operational definitions $\bar{C}(\Phi)\le C(\Phi)\le C_{\mathrm{ea}}(\Phi)$. Analytically this follows (by means of~\eqref{chi-cap}
and \eqref{ea-cap}) from the following expression for
the quantum mutual information:
\begin{equation}\label{mi-rep}
I(\rho, \Phi)=H(\rho)+\chi^{}_\Phi(\rho)-\chi^{}_{\widehat{\Phi}}(\rho)
=\chi^{}_\Phi(\rho)+\Delta_\Phi(\rho),
\end{equation}
where $\Delta_\Phi(\rho)=H(\rho)-\chi^{}_{\widehat{\Phi}}(\rho)$. This expression is easily derived by using ~\eqref{chi-fun+} and by noting
that
$\hat{H}_\Phi\equiv\hat{H}_{\widehat{\Phi}}$ (this follows from
coincidence of the functions $\rho\mapsto H(\Phi(\rho))$ and
$\rho\mapsto H(\widehat{\Phi}(\rho))$ on the set of pure states).\smallskip

Since $H(\rho)=\sum_i\pi_i H(\rho_i\|\rho)$ for any ensemble $\{\pi_i,\rho_i\}$ of pure states with the average state $\rho$, we
have
\begin{equation}\label{delta-rep}
\Delta_\Phi(\rho)=\min_{\substack{\sum_i\pi_i\rho_i=\rho\\
\rank\rho_i=1}}\,\sum_i\pi_i
\left[H(\rho_i\|\rho)-H(\widehat{\Phi}(\rho_i)\|\widehat{\Phi}(\rho))\right]\ge 0,
\end{equation}
where the last inequality follows from  monotonicity of the relative
entropy.\medskip

\begin{remark}\label{opt}
The minimum in~\eqref{delta-rep} is achieved at an ensemble $\{\pi_i,\rho_i\}$ of pure states if and only if the maximum in~\eqref{chi-fun} is achieved at this ensemble. Indeed, since $\sum_i\pi_i
H(\Phi(\rho_i))=\sum_i\pi_i H(\widehat{\Phi}(\rho_i))$,
this can be easily shown by using expression ~\eqref{chi-fun+} for
the $\chi$-functions of the channels $\Phi$ and $\widehat{\Phi}$.
\end{remark}

\subsection{General inequalities}

Expression~\eqref{mi-rep} immediately implies the general upper
bound
$$
C_{\mathrm{ea}}(\Phi)\le\bar{C}(\Phi)+\log\dim\mathcal{H}_A,
$$
proved in~\cite{Fan+,H-r-c-c} by different methods. By using this
expression and by noting that
$\chi^{}_\Phi(\rho)-\chi^{}_{\widehat{\Phi}}(\rho)=I_c(\rho,\Phi)$ is the coherent information of the channel $\Phi$ at the state
$\rho$ (see~\cite{Sch}) it easy to obtain the following
inequalities\footnote{Here and in what follows the subscription in the third inequality means that it holds under the condition
$H(\Phi(\rho))\ge H(\rho)$ for all
$\rho\in\mathfrak{S}(\mathcal{H}_A)$. This condition is valid, in particular, for all bistochastic channels.}:
\begin{equation}\label{d-in}
\!\!\begin{array}{l}
H(\rho_1)-\bar{C}(\widehat{\Phi})\le C_{\mathrm{ea}}(\Phi)-\bar{C}(\Phi) \\\\ \le
H(\rho_2)-\chi^{}_{\widehat{\Phi}}(\rho_2)\mathrel{\underset{H(\Phi(\cdot))\ge
H(\cdot)}{\le}}H(\Phi(\rho_2))-\chi^{}_{\widehat{\Phi}}(\rho_2) =
I_c(\rho_2,\Phi)+\hat{H}_\Phi(\rho_2),
\end{array}\!\!\!
\end{equation}
where $\rho_1$ and $\rho_2$ are states in $\mathfrak{S}(\mathcal{H}_A)$ such that
$\chi^{}_\Phi(\rho_1)=\bar{C}(\Phi)$ (i.e. $\rho_1$ is the average
state of an optimal ensemble) and
$I(\rho_2,\Phi)=C_{\mathrm{ea}}(\Phi)$.\medskip

Let $Q_1(\Phi)=\max_{\rho\in\mathfrak{S}(\mathcal{H}_A)}I_c(\rho, \Phi)$
and $Q(\Phi)=\lim_{n\to+\infty} n^{-1}Q_1(\Phi^{\otimes n})$ be the quantum capacity of the channel
$\Phi$~\cite{N&Ch}. The following proposition contains several  estimations derived from~\eqref{d-in}.
\medskip

\begin{property}\label{simple}
\emph{Let \/ $\Phi\colon\mathfrak{S}(\mathcal{H}_A)\to\mathfrak{S}(\mathcal{H}_B)$ be a quantum
channel and
\/ $\widehat{\Phi}\colon\mathfrak{S}(\mathcal{H}_A)\to\mathfrak{S}(\mathcal{H}_E)$
its
complementary channel.}
\medskip

A) \emph{The following inequalities hold
\begin{equation}
\!\bar{C}(\Phi)-\bar{C}(\widehat{\Phi})\le
C_{\mathrm{ea}}(\Phi)-\bar{C}(\Phi)\mathrel{\underset{H(\Phi(\cdot))\ge
H(\cdot)}{\le}} Q_1(\Phi)+\min\sum_i\pi_iH(\Phi(\rho_i)),\!\label{one}
\end{equation}
\begin{equation}
\!C(\Phi)-C(\widehat{\Phi})\le
C_{\mathrm{ea}}(\Phi)-C(\Phi)\mathrel{\underset{H(\Phi(\cdot))\ge H(\cdot)}{\le}}
Q(\Phi)+\min\sum_i\pi_iH(\Phi(\rho_i)),\!\label{two}
\end{equation}
where the minimum is over all ensembles $\{\pi_i,\rho_i\}$ of pure states such that  $I\Bigl(\sum_i\pi_i\rho_i,
\Phi\Bigr)=C_{\mathrm{ea}}(\Phi)$. This term  can be replaced by \/ $\max_{\rho\in\,\mathop{\rm
extr}\mathfrak{S}(\mathcal{H}_A)}H(\Phi(\rho))$.}
\medskip

B) \emph{If the average state of at least one optimal ensemble for
the channel~$\,\Phi$ coincides with the chaotic state \/ $\rho_c=(\dim\mathcal{H}_A)^{-1}I_A$ then
$$
C_{\mathrm{ea}}(\Phi)-\bar{C}(\Phi)\ge
\log\dim\mathcal{H}_A-\bar{C}(\widehat{\Phi})
$$
and hence $\bar{C}(\Phi)=C_{\mathrm{ea}}(\Phi)\;\Rightarrow\;\bar{C}(\widehat{\Phi})
=\log\dim\mathcal{H}_A$.\footnote{Note that
$\bar{C}(\widehat{\Phi})\le\log\dim\mathcal{H}_A$ for any channel $\Phi$.}}
\medskip

C) \emph{If $C_{\mathrm{ea}}(\Phi)=I(\rho_c,\Phi)$ then
$\,\bar{C}(\widehat{\Phi})=\log\dim\mathcal{H}_A\;\Rightarrow\;\bar{C}(\Phi)
=C_{\mathrm{ea}}(\Phi)$. If, in addition, the average state of at least one optimal ensemble
for the channel~$\,\widehat{\Phi}$ coincides with the chaotic state $\rho_c$ then}
$$
C_{\mathrm{ea}}(\Phi)-\bar{C}(\Phi)\le
\log\dim\mathcal{H}_A-\bar{C}(\widehat{\Phi}).
$$
\end{property}

\begin{proof}
A) Inequality~\eqref{one}  directly follows from~\eqref{d-in}. To obtain
inequality \eqref{two} by regularization from~\eqref{d-in} it is sufficient  to note that the function\break
$\mathfrak{S}(\mathcal{H}^{\otimes n}_A)\ni\omega\mapsto I(\omega,\Phi^{\otimes n})$
attains maximum at
the state $\rho^{\otimes n}_2$ by subadditivity of the quantum mutual information and to use the obvious inequality $\hat{H}_{\Phi^{\otimes n}}(\rho_2^{\otimes n})\le
n\hat{H}_\Phi(\rho_2)$.\medskip

B) This assertion directly follows from inequality~\eqref{d-in}.\medskip

C) To derive the first part of this assertion from inequality~\eqref{d-in} note that $\bar{C}(\widehat{\Phi})=\log\dim\mathcal{H}_A$
implies $\bar{C}(\widehat{\Phi})=\chi^{}_{\widehat{\Phi}}(\rho_c)$. The
second part directly follows from the second inequality in~\eqref{d-in}.
\end{proof}

\begin{remark}\label{simple-r}
Since $\bar{C}(\widehat{\Phi})\le\log\dim\mathcal{H}_E$, we have
$$
C_{\mathrm{ea}}(\Phi)-\bar{C}(\Phi)\ge \log\dim\mathcal{H}_A-\log\dim\mathcal{H}_E
$$
for any channel $\Phi$ satisfying the condition of Proposition~\ref{simple},\,B) and hence $C_{\mathrm{ea}}(\Phi)>\bar{C}(\Phi)$ if the dimension of the
environment (=the minimal number of Kraus operators) is less than
the dimension of the input space of the channel $\Phi$.\medskip

For an arbitrary channel $\Phi$ inequality~\eqref{d-in} implies
$$
C_{\mathrm{ea}}(\Phi)-\bar{C}(\Phi)\ge
H(\bar{\rho})-\log\dim\mathcal{H}_E\ge\bar{C}(\Phi)-\log\dim\mathcal{H}_E,
$$
where $\bar{\rho}$ is the average state of any optimal ensemble for the channel $\Phi$.
\end{remark}

\subsection{Conditions for the equality $\bar{C}(\Phi)=C_{\mathrm{ea}}(\Phi)$ based on the Petz theorem}

By using expressions~\eqref{mi-rep} and~\eqref{delta-rep}, monotonicity of the relative entropy and the Petz theorem \cite[Theorem~3]{Petz}
characterizing the case in which
monotonicity of the relative entropy holds with an equality, the
following necessary and sufficient conditions for the equality $\bar{C}(\Phi)=C_{\mathrm{ea}}(\Phi)$ can be obtained.\medskip

\begin{theorem}\label{coin-crit}
\emph{Let $\,\Phi\colon\mathfrak{S}(\mathcal{H}_A)\to\mathfrak{S}(\mathcal{H}_B)$ be a quantum
channel and\break
$\,\widehat{\Phi}\colon\mathfrak{S}(\mathcal{H}_A) \to\mathfrak{S}(\mathcal{H}_E)$
its complementary channel.}\medskip

A) \emph{If there exist a
channel \/ $\Theta\colon\mathfrak{S}(\mathcal{H}_E)\to\mathfrak{S}(\mathcal{H}_A)$
and an ensemble
\/ $\{\pi_i,\rho_i\}$ of pure states such that
\begin{equation}\label{inv-cond}
\Theta(\widehat{\Phi}(\rho_i))=\rho_i,\quad \forall i,
\end{equation}
and $I(\bar{\rho}, \Phi)=C_{\mathrm{ea}}(\Phi)$, where \/
$\bar{\rho}=\sum_i\pi_i\rho_i$, then
$\bar{C}(\Phi)=C_{\mathrm{ea}}(\Phi)$.\footnote{It is sufficient to require that $\Theta$ is a trace preserving positive map for which monotonicity of the relative entropy holds.}}\medskip

B) \emph{If $\bar{C}(\Phi)=C_{\mathrm{ea}}(\Phi)$ then for an
arbitrary optimal ensemble \/ $\{\pi_i,\rho_i\}$ of pure states for the
channel \/ $\Phi$  with the average state
$\bar{\rho}$ there exists a channel \/
$\Theta\colon\mathfrak{S}(\mathcal{H}_E)\to\mathfrak{S}(\mathcal{H}_A)$ such that
\/~\eqref{inv-cond} holds. The channel \/ $\Theta$ can be defined by means of an arbitrary non-degenerate probability distribution \/ $\{\hat{\pi}_i\}$ by setting its action on any state \/ $\sigma$ supported by the subspace
$\,\supp\widehat{\Phi}(\bar{\rho})$ as follows
\begin{equation}\label{th-f}
\Theta(\sigma)=[\hat{\rho}]^{1/2}\widehat{\Phi}^*
\left(\bigl[\widehat{\Phi}(\hat{\rho})\bigr]^{-1/2}\sigma
\bigl[\widehat{\Phi}(\hat{\rho})\bigr]^{-1/2}\right)[\hat{\rho}]^{1/2},
\end{equation}
where  $\hat{\rho}=\sum_i\hat{\pi}_i\rho_i$ and \/ $\widehat{\Phi}^*$ is a dual map to the channel \/
$\widehat{\Phi}$.}

\emph{If \/ $\{\hat{\pi}_i\}$ is a degenerate probability distribution then relation ~\eqref{inv-cond} holds for the channel  \/ $\Theta$ defined by~\eqref{th-f} for all $i$ such that \/ $\hat{\pi}_i>0$.}
\end{theorem}

\begin{proof}
A) If $\{\pi_i,\rho_i\}$ is an ensemble of pure states with the
average state $\bar{\rho}$ for which~\eqref{inv-cond} holds then
monotonicity of the relative entropy and~\eqref{delta-rep} imply
$\Delta_\Phi(\bar{\rho})=0$ and hence $C_{\mathrm{ea}}(\Phi)=I(\bar{\rho},
\Phi)=\chi^{}_\Phi(\bar{\rho})\le\bar{C}(\Phi)$.\medskip

B) Since $\chi^{}_\Phi(\rho)\le I(\rho, \Phi)$ for any state $\rho$ by \eqref{mi-rep}, it is easy to see
that
$\bar{C}(\Phi)=C_{\mathrm{ea}}(\Phi)$ implies $\chi^{}_\Phi(\bar{\rho})=I(\bar{\rho}, \Phi)$
for any an optimal
ensemble $\{\pi_i,\rho_i\}$ of pure states with the average state $\bar{\rho}$. It follows from \eqref{delta-rep} and Remark~\ref{opt} that
$$
H(\rho_i\|\bar{\rho})=H(\widehat{\Phi}(\rho_i)\|\widehat{\Phi}(\bar{\rho})),\quad
\forall i.
$$
Hence the Petz theorem~\cite[Theorem 3]{Petz} implies existence of the
channel $\Theta$ for which~\eqref{inv-cond} holds. By monotonicity of the relative entropy for arbitrary probability distribution  $\{\hat{\pi}_i\}$ we have
$$
H(\rho_i\|\hat{\rho})=H(\widehat{\Phi}(\rho_i)\|\widehat{\Phi}(\hat{\rho})),\quad
\hat{\rho}=\sum_i\hat{\pi}_i\rho_i,
$$
for all $i$ such that $\hat{\pi}_i>0$. Hence the formula for
the channel $\Theta$ also follows from the Petz theorem.
\end{proof}

Theorem~\ref{coin-crit},\,A) makes it possible to prove the equality $C_{\mathrm{ea}}(\Phi)=\bar{C}(\Phi)$ for all classical-quantum channels (see Theorem~\ref{e-b-ch} in Section~2.3).\medskip

Theorem~\ref{coin-crit},\,B) can be used to prove the strict inequality $C_{\mathrm{ea}}(\Phi)>\bar{C}(\Phi)$,
by showing that~\eqref{inv-cond} can not be valid for an optimal ensemble
$\{\pi_i,\rho_i\}$ and the channel $\Theta$ defined by~\eqref{th-f}.\medskip

\begin{example}\label{measurement}
Consider the entanglement-breaking  channel
$$
\Phi(\rho)=\sum_k\langle\varphi_k|\rho|\varphi_k\rangle|k\rangle\langle k|,
$$
where $\{|\varphi_k\rangle\}$ is an overcomplete system of vectors
in the space $\mathcal{H}_A$ (that is
$\sum_k|\varphi_k\rangle\langle\varphi_k|=I_A$) and $\{|k\rangle\}$ is an orthonormal basis in the space
$\mathcal{H}_B$. It is
easy to see that $\Phi=\widehat{\Phi}$. Hence $I(\rho,\Phi)=H(\rho)$ and
$C_{\mathrm{ea}}(\Phi)=\log\dim\mathcal{H}_A$.

Suppose that $\bar{C}(\Phi)=C_{\mathrm{ea}}(\Phi)=\log\dim\mathcal{H}_A$. Then the
average state of any optimal ensemble $\{\pi_i,\rho_i\}$ for the
channel $\Phi$ coincides with the chaotic state $\rho_c$ in $\mathfrak{S}(\mathcal{H}_A)$. Since
$\widehat{\Phi}^*(A)=\sum_k\langle k|A|k\rangle|\varphi_k\rangle\langle
\varphi_k|$ and $\widehat{\Phi}(\rho_c)=\Phi(\rho_c)$ is a full rank state, relation~\eqref{inv-cond}
can be valid for the channel $\Theta$ defined by
\eqref{th-f} only if $\rho_i=|\varphi_{k_i}\rangle\langle \varphi_{k_i}|$ for some $k_i$ and
$$\rank\widehat{\Phi}(|\varphi_{k_i}\rangle\langle
\varphi_{k_i}|)=\rank\sum_k\langle\varphi_k|\varphi_{k_i}\rangle\langle
\varphi_{k_i}|\varphi_k\rangle|k\rangle\langle k|=1$$ for all $i$. But
this can be valid only if $\{|\varphi_k\rangle\}$ is an orthonormal
basis. So,  we conclude that
$$
C_{\mathrm{ea}}(\Phi)=\bar{C}(\Phi)\quad\Leftrightarrow\quad\{|\varphi_k\rangle\}\
\text{is an ortonormal basis}.
$$
\end{example}

The same conclusion was obtained in~\cite{H-CQM} as a corollary of a
general criterion for the equality $C_{\mathrm{ea}}(\Phi)=\bar{C}(\Phi)$ for the
class of channels defined by quantum observables, which is proved by
means of the ensemble-measurement duality.

\subsection{A simple criterion for the equality $\bar{C}(\Phi)=C_{\mathrm{ea}}(\Phi)$.}

Now we will show that the equality $\bar{C}(\Phi)=C_{\mathrm{ea}}(\Phi)$ holds if (correspondingly, only if) the channel $\Phi$ (correspondingly, the subchannel of $\Phi$ determining its classical capacity) belongs to the class of classical-quantum channels.\smallskip

A channel $\Phi\colon\mathfrak{S}(\mathcal{H}_A)\to\mathfrak{S}(\mathcal{H}_B)$ is called \emph{classical-quantum} if it has the following representation
\begin{equation}\label{c-q-rep}
\Phi(\rho)=\sum_{k=1}^{\dim\mathcal{H}_A}\langle
k|\rho|k\rangle\sigma_k,\quad\rho\in\mathfrak{S}(\mathcal{H}_A),
\end{equation}
where $\{|k\rangle\}$ is an orthonormal basis in $\mathcal{H}_A$ and $\{\sigma_k\}$ is a collection of states in  $\mathfrak{S}(\mathcal{H}_B)$
\cite{e-b-ch,N&Ch}.\medskip

For correct formulation of the above statement we will need the following notion.\smallskip

\begin{definition}\label{ef-subch}
Let $\mathcal{H}^\chi_\Phi$ be the minimal subspace of $\mathcal{H}_A$ containing elements of
all optimal ensembles for the channel
$\Phi\colon\mathfrak{S}(\mathcal{H}_A)\to\mathfrak{S}(\mathcal{H}_B)$. The restriction $\Phi_{\chi}$
of the channel $\Phi$ to the set
$\mathfrak{S}(\mathcal{H}^\chi_\Phi)$ is called \emph{$\chi$-essential\/} part (subchannel) of the channel~$\Phi$.
\end{definition}\medskip

If $\mathcal{H}^\chi_\Phi\ne\mathcal{H}_A$ then pure states corresponding to vectors in
$\mathcal{H}_A\setminus\mathcal{H}^\chi_\Phi$ can not be used as elements of optimal
ensemble for the channel $\Phi$. This means, roughly speaking,
that these states are useless for non-entangled coding of classical information and
hence it is natural to consider the $\chi$-essential subchannel
$\Phi_{\chi}$ instead of the channel $\,\Phi$ dealing with the
Holevo capacity of the channel $\Phi$ (which coincides with the classical capacity if
$C_{\mathrm{ea}}(\Phi)=\bar{C}(\Phi)$).

By definition $\bar{C}(\Phi_{\chi})=\bar{C}(\Phi)$. Hence $C_{\mathrm{ea}}(\Phi)=\bar{C}(\Phi)$ implies
$C_{\mathrm{ea}}(\Phi_{\chi})=C_{\mathrm{ea}}(\Phi)$. Thus, in this case speaking about the entanglement-assisted capacity of the
channel $\Phi$ we  may also consider the $\chi$-essential subchannel
$\Phi_{\chi}$ instead of the channel $\Phi$.\medskip

Theorem~\ref{coin-crit} makes it possible to prove the following assertions.\medskip

\begin{theorem}\label{e-b-ch}
\emph{Let $\,\Phi\colon\mathfrak{S}(\mathcal{H}_A)\to\mathfrak{S}(\mathcal{H}_B)$ be a quantum channel.}
\medskip

 A) \emph{If $\,\Phi$ is a classical-quantum channel then $\,C_{\mathrm{ea}}(\Phi)=\bar{C}(\Phi)$.}\medskip

B) \emph{If  $C_{\mathrm{ea}}(\Phi)=\bar{C}(\Phi)$ then the $\chi$-essential part
of the channel $\Phi$ is a classical-quantum channel.}
\end{theorem}\medskip

Example~\ref{n-e-b-ch} below shows that in general the $\chi$-essential part of the channel $\Phi$ in Theorem~\ref{e-b-ch},\,B) can not replaced by the channel $\Phi$.

\begin{proof}
A) If the channel $\Phi$ has representation~\eqref{c-q-rep} then $\Phi=\Phi\circ\Pi$, where $\Pi(\rho)=\sum_k\langle
k|\rho|k\rangle|k\rangle\langle k|$ is a channel from $\mathfrak{S}(\mathcal{H}_A)$ to itself.

It is easy to show (see~\cite[the proof of Lemma~17]{R&Co}) existence of a channel  $\Theta$ such that
$\Theta\circ\widehat{\Phi\circ\Pi}=\widehat{\Pi}=\Pi$.

By the chain rule for the quantum mutual information (see~\cite{N&Ch}) we have
$$
I(\rho,\Phi)=I(\rho,\Phi\circ\Pi)\le I(\Pi(\rho),\Phi).
$$
It follows that the function $\rho\mapsto I(\rho,\Phi)$ attains
maximum at a state diagonizable in the basis $\{|k\rangle\}$. Since
$\Theta\circ\widehat{\Phi\circ\Pi}(|k\rangle\langle k|)=\Pi(|k\rangle\langle k|)
=|k\rangle\langle k|$ for any $k$, Theorem~\ref{coin-crit},\,A)
implies $C_{\mathrm{ea}}(\Phi)=\bar{C}(\Phi)$.\medskip

B) Replacing the channel $\Phi$ by its $\chi$-essential subchannel, we may consider that $\mathcal{H}^\chi_\Phi=\mathcal{H}_A$.

Let $\Phi(\rho)=\sum_{i=1}^nV_i\rho V_i^*$ be a minimal Kraus
representation of the channel $\Phi$. Then
$$
\widehat{\Phi}(\rho)=\sum_{i,j=1}^n \Tr V_i\rho V_j^*|i\rangle\langle
j|\qquad\textrm{and}\qquad\widehat{\Phi}^*(A)=\sum_{i,j=1}^n \langle j |A|i\rangle
V_j^*V_i,
$$
where $\left\{|i\rangle\right\}_{i=1}^n$ is an orthonormal basis in the
$n$-dimensional Hilbert space $\mathcal{H}_E$.

Let $\{\pi_k,|\varphi_k\rangle\langle \varphi_k|\}$ be an optimal ensemble of pure states
for the channel $\Phi$ with a full rank average state. We may assume that
$\left\{|\varphi_k\rangle\right\}_{k=1}^m$, $m=\dim\mathcal{H}_A$, is a basis in the space
$\mathcal{H}_A$. Let $\hat{\pi}_k=1/m$, $k=\overline{1,m}$. Then
$\hat{\rho}=\sum_{k=1}^m\hat{\pi}_k|\varphi_k\rangle\langle \varphi_k|$ is a full rank state in
$\mathfrak{S}(\mathcal{H}_A)$. Since $\mathcal{H}_E$ is an environment space of minimal
dimension, $\widehat{\Phi}(\hat{\rho})$ is a full rank state in
$\mathfrak{S}(\mathcal{H}_E)$.

Let $|\phi_k\rangle=\!\sqrt{\hat{\pi}_k\hat{\rho}^{-1}}|\varphi_k\rangle$ and
$B_k=\hat{\pi}_k\bigl[\widehat{\Phi}(\hat{\rho})\bigr]^{-1/2}
\widehat{\Phi}(|\varphi_k\rangle\langle
\varphi_k|)\bigl[\widehat{\Phi}(\hat{\rho})\bigr]^{-1/2}$, $k=\overline{1,m}$. Since
$\sum_{k=1}^m|\phi_k\rangle\langle \phi_k|=I_{\mathcal{H}_A}$,
$\left\{|\phi_k\rangle\right\}_{k=1}^m$ is an orthonormal basis in $\mathcal{H}_A$.
By Theorem~\ref{coin-crit},\,B) $|\phi_k\rangle\langle \phi_k|=\widehat{\Phi}^*(B_k)$ for all
$k$. By
the spectral theorem  $B_k=\sum_p|\psi^p_k\rangle\langle \psi^p_k|$, where
$\left\{|\psi^p_k\rangle\right\}_p$ is a set of vectors in $\mathcal{H}_E$, for
each $k$. Since
$\widehat{\Phi}(\hat{\rho})$~is a full rank state,
we have
$$
\sum_{k,p}|\psi^p_k\rangle\langle \psi^p_k|=\sum_k B_k=I_E.
$$
By Lemma~\ref{new-kraus-rep} below  $\Phi(\rho)=\sum_{k,p}W_{kp}\rho W_{kp}^*$, where
$W_{kp}=\sum^n_{i=1}\langle\psi^p_k|i\rangle V_i$.

Since $|\phi_k\rangle\langle \phi_k|
=\widehat{\Phi}^*\Bigl(\sum_p|\psi^p_k\rangle\langle \psi^p_k|\Bigr)$ for each $k$ and
$$
\widehat{\Phi}^*(|\psi^p_k\rangle\langle \psi^p_k|)=\sum_{i,j=1}^n \langle j
|\psi^p_k\rangle\langle\psi^p_k|i\rangle V_j^*V_i=W_{kp}^*W_{kp},
$$
there exists a collection $\{|\beta_{kp}\rangle\}$ of vectors in $\mathcal{H}_B$ such that $W_{kp}=|\beta_{kp}\rangle\langle
\phi_k|$ and $\sum_p\|\beta_{kp}\|^2=1$ for each $k$. Hence
$$
\Phi(\rho)=\sum_{k,p}W_{kp}\rho
W_{kp}^*=\sum_k\langle\phi_k|\rho|\phi_k\rangle\sum_{p}|\beta_{kp}\rangle\langle
\beta_{kp}|.
$$
\end{proof}

\begin{lemma}\label{new-kraus-rep}
\emph{Let \/ $\Phi(\rho)=\sum_{i=1}^nV_i\rho V_i^*$ be a quantum channel and  \/
$\left\{|i\rangle\right\}_{i=1}^n$ be an orthonormal basis in the $n$-dimensional Hilbert space \/ $\mathcal{H}_E$.
An arbitrary overcomplete
system \/
$\left\{|\psi_k\rangle\right\}_k$ of vectors in \/ $\mathcal{H}_E$ generates the
Kraus representation \/
$\Phi(\rho)=\sum_kW_k\rho W_k^*$ of the
channel \/ $\Phi$, where $W_k=\sum^n_{i=1}\langle\psi_k|i\rangle V_i$.}
\end{lemma}

\begin{proof}
Since $\sum_k|\psi_k\rangle\langle \psi_k|=I_E$\rule[-7pt]{0pt}{1pt}, we have
$$
\begin{aligned}
\sum_kW_k\rho W^*_k&=\sum^n_{i,j=1} V_i\rho V^*_j \sum_k\langle\psi_k|i\rangle\langle
j|\psi_k\rangle\\ & =\sum^n_{i,j=1} V_i\rho V^*_j \sum_k\Tr |i\rangle\langle
j||\psi_k\rangle\langle\psi_k|&=\sum^n_{i=1} V_i\rho V^*_i.
\end{aligned}
$$
\end{proof}

\begin{remark}\label{q-c-ch}
The assertions of Theorem~\ref{e-b-ch} agree with the obtained in~\cite{H-CQM} criterion for the equality $C_{\mathrm{ea}}(\Phi)=\bar{C}(\Phi)$ for the quantum-classical channel
$$
\Phi(\rho)=\sum_k[\Tr M_k\rho]|k\rangle\langle k|
$$
defined by the collection $\{M_k\}$ of positive operators in $\mathcal{H}_A$ such that $\sum_kM_k=I_A$, where $\{|k\rangle\}$ is an orthonormal basis in $\mathcal{H}_B$. Indeed, it is easy to see that this channel is classical-quantum if and only if \/$M_kM_l=M_lM_k$ for all $k,l$.
\end{remark}\medskip

Since $\mathcal{H}^\chi_\Phi=\mathcal{H}_A$ means existence of an optimal ensemble for the channel $\Phi$ with a full rank average state, Theorem~\ref{e-b-ch} implies the following criterion for coincidence of the capacities.
\medskip

\begin{corollary}\label{e-b-ch-c}
\emph{Let \/ $\Phi$ be a quantum channel
for which there exists an optimal ensemble with a full rank
average state. Then}
$$
C_{\mathrm{ea}}(\Phi)=\bar{C}(\Phi)\quad\Leftrightarrow\quad\Phi\ \textit{is a classical-quantum channel}.
$$
\end{corollary}

The following example proposed in~\cite{BSST+} (as an example of non--entanglement-breaking channel such that $C_{\mathrm{ea}}(\Phi)=\bar{C}(\Phi)$) shows that the "full rank average state" condition in Corollary~\ref{e-b-ch-c} is essential.\medskip

\begin{example}\label{n-e-b-ch}
Let $\mathcal{H}_1$, $\mathcal{H}_2$ and $\mathcal{H}_3$ be qubit spaces. Let
$\left\{|k\rangle\right\}_{k=1}^4$ and $\{|-\rangle, |+\rangle\}$ be orthonormal bases
in $\mathcal{K}=\mathcal{H}_1\otimes\mathcal{H}_2$ and in $\mathcal{H}_3$ correspondingly. Consider the channel
$$
\Phi(\rho)=\sum_{k=1}^4\left[\langle
k|\otimes\langle+|\right]\rho\left[|k\rangle\otimes|+\rangle\right]|k\rangle\langle
k|+\textstyle\frac{1}{2}\displaystyle I_{\mathcal{H}_2}
\otimes\Tr_{\mathcal{H}_2\otimes\mathcal{H}_3}[I_{\mathcal{K}}
\otimes|-\rangle\langle -|]\rho
$$
from $\mathfrak{S}(\mathcal{K}\otimes\mathcal{H}_3)$ into $\mathfrak{S}(\mathcal{K})$. It is easy to show that
$C_{\mathrm{ea}}(\Phi)=\bar{C}(\Phi)=2$ and $Q(\Phi)=1$ \cite{BSST+}. Thus the channel $\Phi$ is non-entanglement-breaking and hence it is not classical-quantum.

Since $\bar{C}(\Phi)=2=\log\dim\mathcal{K}$, any optimal ensemble for the
channel $\Phi$ can not contain states with nonzero output entropy.
Thus the subspace $\mathcal{H}_\Phi^\chi$ consists of vectors
$|\varphi\rangle\otimes|+\rangle$, $|\varphi\rangle\in\mathcal{K}$. Hence the $\chi$-essential part of the channel
$\Phi$ is isomorphic to the classical-quantum channel
$\rho\mapsto\sum_{k=1}^4\langle k|\rho|k\rangle|k\rangle\langle k|$ (in accordance with Theorem~\ref{e-b-ch},\,B)).
\end{example}

\subsection{On covariant channels}

The class of channels, for which the conditions of the parts B and C
of Proposition~\ref{simple} and of Corollary \ref{e-b-ch-c} hold simultaneously, contains any
channel $\Phi$  \emph{\/covariant} with respect to representations $\left\{V_g\right\}_{g\in G}$ and
$\left\{W_g\right\}_{g\in G}$ of a compact group~$G$ in
the sense that
\begin{equation}\label{cov}
\Phi(V_g\rho V^*_g)=W_g\Phi(\rho)W^*_g,\quad\forall g\in G,
\end{equation}
provided the representation $\left\{V_g\right\}_{g\in G}$ is irreducible.
Indeed, irreducibility of the representation $\left\{V_g\right\}_{g\in G}$
implies
\begin{equation}\label{mixt}
\rho_c\doteq(\dim\mathcal{H}_A)^{-1}I_A=\int_G V_g\rho V^*_g
\,\mu_H(dg),\quad \forall\rho\in\mathfrak{S}(\mathcal{H}_A),
\end{equation}
where $\mu_H$ is the Haar measure on the group $G$~\cite{H-r-c-c}. So, to prove that
\begin{equation}\label{coin}
\bar{C}(\Phi)=\chi^{}_\Phi(\rho_c),
\quad\bar{C}(\widehat{\Phi})=\chi^{}_{\widehat{\Phi}}(\rho_c),\quad
C_{\mathrm{ea}}(\Phi)=I(\rho_c,\Phi)
\end{equation}
it is sufficient, by concavity of the $\chi$-function and of the
quantum mutual information, to show that
\begin{equation}\label{inv}
\chi^{}_\Phi(\rho)=\chi^{}_\Phi(V_g\rho V^*_g),\;\;
\chi^{}_{\widehat{\Phi}}(\rho)=\chi^{}_{\widehat{\Phi}}(V_g\rho
V^*_g),\;\; I(\rho,\Phi)=I(V_g\rho V^*_g,\Phi)
\end{equation}
for all $g\in G$ and $\rho\in\mathfrak{S}(\mathcal{H}_A)$.

The first and the third  equalities in~\eqref{inv} can be easily
proved by using~\eqref{chi-fun} and the well known expression for
the quantum mutual information via the relative entropy (by means of
invariance of the relative entropy with respect to unitary
transformations of the both their arguments). By these equalities the
second one follows from \eqref{mi-rep}.

The class of covariant channels is sufficiently large, it contains
all unital qubit channels and nontrivial classes of channels in
higher dimensions~\cite{H&F,H-r-c-c}.

By using~\eqref{mixt} and~\eqref{coin} it is easy to show that
(cf.\cite{H-r-c-c})
\begin{equation}\label{cap-exp}
\begin{gathered}
\bar{C}(\Phi)=H(\Phi(\rho_c))-H_{\min}(\Phi),
\qquad\bar{C}(\widehat{\Phi})=H(\widehat{\Phi}(\rho_c))-H_{\min}(\Phi),\\\\
C_{\mathrm{ea}}(\Phi)=\log\dim\mathcal{H}_A+H(\Phi(\rho_c))
-H(\widehat{\Phi}(\rho_c))
\end{gathered}
\end{equation}
for any channel $\Phi\colon\mathfrak{S}(\mathcal{H}_A)\to\mathfrak{S}(\mathcal{H}_B)$ satisfying the
above covariance condition, where
$H_{\min}(\Phi)=\min_{\rho\in\mathfrak{S}(\mathcal{H}_A)}H(\Phi(\rho))$ is the minimal output entropy of the
channel $\Phi$ (coinciding with $H_{\min}(\widehat{\Phi})$). If, in addition, the
representation $\left\{W_g\right\}_{g\in G}$ is also irreducible then
$H(\Phi(\rho_c))$ in~\eqref{cap-exp} can be replaced by $\log\dim\mathcal{H}_B$~\cite{H-r-c-c}.\smallskip

Let $Q_1(\Phi)=\max_{\rho\in\mathfrak{S}(\mathcal{H}_A)}I_c(\rho, \Phi)$ and
$Q(\Phi)=\lim_{n\to+\infty} n^{-1}Q_1(\Phi^{\otimes n})$ be the quantum capacity of the channel $\Phi$. By the above
observations Proposition~\ref{simple} and Corollary~\ref{e-b-ch-c}
imply the following assertions.\smallskip

\begin{property}\label{covariant}
\emph{Let
$\,\Phi\colon\mathfrak{S}(\mathcal{H}_A)\to\mathfrak{S}(\mathcal{H}_B)$ be a channel
satisfying covariance condition~\eqref{cov}. Then}
$$
C_{\mathrm{ea}}(\Phi)=\bar{C}(\Phi)\quad\Leftrightarrow\quad \Phi\ \textit{is a classical-quantum channel}.
$$

\emph{If, in addition, $\dim\mathcal{H}_B\ge\dim\mathcal{H}_A$ and the representation \/ $\left\{W_g\right\}_{g\in
G}$ is irreducible
then
}$$
\begin{gathered}
C_{\mathrm{ea}}(\Phi)-\bar{C}(\Phi)=
\log\dim\mathcal{H}_A-\bar{C}(\widehat{\Phi})\le Q_1(\Phi)+H_{\min}(\Phi),\\
C_{\mathrm{ea}}(\Phi)-C(\Phi)=\log\dim\mathcal{H}_A-C(\widehat{\Phi})\le
Q(\Phi)+H_{\min}(\Phi).
\end{gathered}
$$
\end{property}

\begin{proof}
If the representation $\left\{W_g\right\}_{g\in G}$ is
irreducible then it is easy to show that $\Phi((\dim\mathcal{H}_A)^{-1}I_A)
=(\dim\mathcal{H}_B)^{-1}I_B$~\cite{H-r-c-c}. This and the condition $\dim\mathcal{H}_B\ge\dim\mathcal{H}_A$ imply
$H(\Phi(\rho))\ge H(\rho)$ for any $\rho\in\mathfrak{S}(\mathcal{H}_A)$ by monotonicity
of the relative entropy. Coincidence of the last term in~\eqref{one} and
\eqref{two} with~$H_{\min}(\Phi)$ follows from~\eqref{mixt} and~\eqref{coin}.
\end{proof}

\subsection{On degradable and anti-degradable channels}

Expression~\eqref{mi-rep} and the chain rule for the $\chi$-function
(i.e. $\chi^{}_{\Psi\circ\Phi}\le\chi^{}_\Phi$) show that
\begin{equation}\label{d-b}
C_{\mathrm{ea}}(\Phi_1)\le\log\dim\mathcal{H}_A\le C_{\mathrm{ea}}(\Phi_2)
\end{equation}
for any anti-degradable channel $\Phi_1$ and any degradable channel
$\Phi_2$.\footnote{A channel $\Phi$ is called degradable if
$\widehat{\Phi}=\Psi\circ\Phi$ for some channel $\Psi$, a channel
$\Phi$ is called anti-degradable if $\widehat{\Phi}$ is a degradable
channel~\cite{R&Co}.} By using the Petz theorem~\cite[Theorem~3]{Petz} one can show that if the first (correspondingly, the second)
inequality in~\eqref{d-b} holds with an equality then the
anti-degradable channel $\Phi_1$ is degradable (correspondingly, the
degradable channel $\Phi_2$ is anti-degradable).\smallskip

The second inequality in~\eqref{d-b} and Theorem~\ref{e-b-ch}
imply the following assertion.

\begin{property}\label{degradable}
\emph{If \/ $\Phi\colon\mathfrak{S}(\mathcal{H}_A)\to\mathfrak{S}(\mathcal{H}_B)$ is a degradable
channel then one of the following alternatives holds:
\begin{itemize}
\item
$\bar{C}(\Phi)<C_{\mathrm{ea}}(\Phi)$;
\item
$\Phi$ is a classical-quantum channel having the representation
\begin{equation}\label{c-q-rep+}
\Phi(\rho)=\sum_{k=1}^{\dim\mathcal{H}_A}\langle k|\rho|k\rangle\sigma_k,\quad
\rho\in\mathfrak{S}(\mathcal{H}_A),
\end{equation}
where $\{|k\rangle\}$ is an orthonormal basis in $\mathcal{H}_A$ and $\{\sigma_k\}$ is a collection of states in
$\,\mathfrak{S}(\mathcal{H}_B)$ with mutually orthogonal supports.
\end{itemize}}
\end{property}
\smallskip

\begin{proof}
Suppose that $\bar{C}(\Phi)=C_{\mathrm{ea}}(\Phi)$. Since $\bar{C}(\Phi)\le\log\dim\mathcal{H}_A$
for any channel $\Phi$, the second inequality in~\eqref{d-b} shows that $\bar{C}(\Phi)=\log\dim\mathcal{H}_A$ and
hence the average state of any optimal ensemble for the channel
$\Phi$ coincides with the chaotic state in
$\mathfrak{S}(\mathcal{H}_A)$. By Corollary \ref{e-b-ch-c} $\Phi$ is a classical-quantum channel having representation \eqref{c-q-rep+}, in which
$\{|k\rangle\}$ is an orthonormal basis in $\mathcal{H}_A$ and $\{\sigma_k\}$ is a collection of states in $\mathfrak{S}(\mathcal{H}_B)$.
We will show that the supports of these states are mutually orthogonal.

Let $\sigma_k=\sum_{i=1}^{\dim\mathcal{H}_B}|\psi_{ki}\rangle\langle \psi_{ki}|$. Then
$\Phi(\rho)=\sum_{k,i}W_{ki}\rho W^*_{ki}$, where $W_{ki}= |\psi_{ki}\rangle\langle
k|$, and by using the standard representation for a complementary channel~(cf.~\cite{H-comp-ch}) we obtain
$$
\widehat{\Phi}(\rho)=\sum_{k,\,l=1}^{\dim\mathcal{H}_A}\langle
k|\rho|l\rangle|k\rangle\langle l|
\otimes\sum_{i,\,j=1}^{\dim\mathcal{H}_B}\langle\psi_{lj}|
\psi_{ki}\rangle|i\rangle\langle
j|\in\mathfrak{S}(\mathcal{H}_A\otimes\mathcal{H}_B).
$$
Since $\Phi$ is a degradable channel with representation~\eqref{c-q-rep+}, we have $\widehat{\Phi}(|k\rangle\langle
l|)=\Psi\circ\Phi(|k\rangle\langle l|)=0$ for all $k\ne l$. Hence the above expression for the channel $\widehat{\Phi}$ implies
$\langle\psi_{lj}|\psi_{ki}\rangle=0$ for all $i,j$ and all $k\ne l$. It follows that
$\supp\sigma_k\perp\supp\sigma_l$ for all $k\ne l$.
\end{proof}

\section{On channels with linear constraints}

Speaking about different capacities of channels between finite
dimensional quantum systems we can use any states for coding
information. But dealing with real infinite dimensional channels
we have to impose particular constraints on the choice of input
code-states to avoid infinite values of the capacities and to be
consistent with the physical implementation of the process of
information transmission. A typical physically motivated constraint
is defined by the requirement of bounded energy of states used for
coding information. This constraint can be called linear, since it
is determined by the linear inequality
\begin{equation}\label{lc}
\Tr H\rho\le h,\quad h>0,
\end{equation}
where $H$ is a positive operator -- Hamiltonian of the input quantum
system. Operational definitions of the Holevo capacity, the unassisted and
the entanglement-assisted classical capacities of a quantum channel with
linear constraints are given in~\cite{H-c-w-c}, where the
corresponding generalizations of the Holevo-Schumacher-Westmoreland and Bennett-Shor-Smolin-Thapliyal theorems are
proved.

The aim of this section is to study relations between the above capacities of a quantum
channel with linear constraints, in particular, to show that the
question of coincidence of these capacities for a given channel depends on the form of the constraint.

For simplicity we restrict attention to the finite
dimensional case.\footnote{Generalizations to infinite dimensions are considered in the second part of~\cite{Sh}.}

The Holevo capacity of the channel $\Phi$ with constraint~\eqref{lc}
can be defined as follows
$$
\bar{C}(\Phi,H,h)=\max_{\Tr H\rho\le h}\chi^{}_\Phi(\rho),
$$
where $\chi^{}_\Phi$ is the $\chi$-function of the channel $\Phi$ defined in~\eqref{chi-fun}. An ensemble $\{\pi_i,\rho_i\}$ of pure states with the average state $\bar{\rho}$ is called
\emph{\/optimal} for the channel $\Phi$ with constraint~\eqref{lc} if
$$
\bar{C}(\Phi,H,h)=\chi^{}_\Phi(\bar{\rho})=\sum_i\pi_i
H(\Phi(\rho_i)\|\Phi(\bar{\rho}))\qquad \textrm{and}\qquad \Tr H\bar{\rho}\le h.
$$

By the generalized Holevo-Schumacher-Westmoreland theorem \cite[Proposition 3]{H-c-w-c} the
classical capacity of the channel $\Phi$ with constraint \eqref{lc}
can be expressed by the following regularization formula
$$
C(\Phi,H,h)=\lim_{n\to+\infty} n^{-1}\bar{C}(\Phi^{\otimes n}, H_n, nh),
$$
where $H_n=H\otimes I\otimes\ldots\otimes I+I\otimes H\otimes I\otimes\ldots\otimes
I+\ldots+I\otimes\ldots\otimes I\otimes H$ (each of $n$
summands consists of $n$ multiples).\smallskip

By the generalized Bennett-Shor-Smolin-Thapliyal theorem \cite[Proposition 4]{H-c-w-c} the
entanglement-assisted capacity of the channel $\Phi$  with
constraint~\eqref{lc} is determined as follows
$$
C_{\mathrm{ea}}(\Phi,H,h)=\max_{\Tr H\rho\le h}I(\rho, \Phi),
$$
where $I(\rho, \Phi)$ is the quantum mutual information of the
channel $\Phi$ at the state $\rho$ defined after~\eqref{ea-cap}.\smallskip

Almost all the results of Section~2 concerning relations between
the capacities  $\bar{C}(\Phi)$ and $C_{\mathrm{ea}}(\Phi)$ can be
reformulated for the corresponding capacities of a constrained
channel. For example, instead of~\eqref{d-in} we have
$$
\begin{array}{cc}
H(\rho_1)-\bar{C}(\widehat{\Phi},H,h)  \le
C_{\mathrm{ea}}(\Phi,H,h)-\bar{C}(\Phi,H,h)\\\\  \le
H(\rho_2)-\chi^{}_{\widehat{\Phi}}(\rho_2)\mathrel{\underset{H(\Phi(\cdot))\ge
H(\cdot)}{\le}}H(\Phi(\rho_2))-\chi^{}_{\widehat{\Phi}}(\rho_2)=
I_c(\rho_2,\Phi)+\hat{H}_\Phi(\rho_2),
\end{array}
$$
where $\rho_1$ and $\rho_2$ are states in $\mathfrak{S}(\mathcal{H}_A)$ such that $\Tr H\rho_i\le h$, $i=1,2$,
$\chi^{}_\Phi(\rho_1)=\bar{C}(\Phi,H,h)$ and $I(\rho_2,\Phi)=C_{\mathrm{ea}}(\Phi,H,h)$.\smallskip

By repeating the corresponding proofs it is easy to obtain the
following  proposition.

\begin{property}\label{coin-crit+}
\emph{The assertions of Proposition~$\ref{simple}$, Theorem\/~$\ref{coin-crit}$ and Theorem\/~{\ref{e-b-ch},\,B)}
remain valid with \/$\bar{C}(\Phi)$ and \/$C_{\mathrm{ea}}(\Phi)$ replaced
respectively by\break \/$\bar{C}(\Phi,H,h)$ and
\/$C_{\mathrm{ea}}(\Phi,H,h)$ (under the natural definition of the $\chi$-essential part of the channel\/ $\Phi$ with constraint~\eqref{lc}).
\/The assertions of Theorem~{\ref{e-b-ch},\,A)} remains valid under this replacement if the basis  $\{|k\rangle\}$ in representation~\eqref{c-q-rep} of the channel \/ $\Phi$ consists of eigenvectors of the operator $H$.}
\end{property}\medskip

The following example shows that the assertion of Theorem~\ref{e-b-ch},\,A) without the additional condition is not valid for constrained channels.
\medskip

\begin{example}\label{classical-costrained}
Consider the classical-quantum channel
$$
\Pi(\rho)=\sum_k\langle k|\rho|k\rangle|k\rangle\langle k|,
$$
where $\{|k\rangle\}$ is an orthonormal basis in $\mathcal{H}_A=\mathcal{H}_B$. Let $h<(\dim\mathcal{H}_A)^{-1}\Tr
H$.\smallskip

By using  the generalized version of Theorem~\ref{coin-crit} we will
show that \emph{$$C_{\mathrm{ea}}(\Pi,H,h)=\bar{C}(\Pi,H,h)$$
if and only if the operator $H$ is diagonizable in the basis~$\{|k\rangle\}$}.\medskip

Since $\Pi=\widehat{\Pi}$, we have $I(\rho,\Pi)=H(\rho)$ and $C_{\mathrm{ea}}(\Pi,H,h)=\max_{\Tr
H\rho\le h}H(\rho)$. By using
the Lagrange method it is easy to show that the above  maximum is
attained at the unique state $\rho_*=(\Tr\exp(-\lambda H))^{-1}\exp(-\lambda H)$, where $\lambda$ is determined by the
equation $\Tr H\exp(-\lambda H)=h\Tr\exp(-\lambda H)$. If
$C_{\mathrm{ea}}(\Pi,H,h)=\bar{C}(\Pi,H,h)$ then Theorem~\ref{coin-crit} implies existence of an ensemble
$\{\pi_i, \rho_i\}$ of pure states with the average state~$\rho_*$ such that
$$
\rho_i=\rho_*^{1/2}\Pi^*\left([\Pi(\rho_*)]^{-1/2}\Pi(\rho_i)
[\Pi(\rho_*)]^{-1/2}\right) \rho_*^{1/2},\quad \forall i.
$$
Since $\Pi^*=\Pi$ and $\rho_*$ is a full rank state, this equality may be valid only if $\rho_i=|k\rangle\langle k|$ for some $k$. Thus $\{|k\rangle\}$ is a basis of eigenvectors for the state $\rho_*$ and hence for the
operator $H$.

If the operator $H$ is diagonizable in the basis $\{|k\rangle\}$
then $\rho_*=\sum_k\pi_k|k\rangle\langle k|$ and hence
$$
\bar{C}(\Pi,H,h)\ge\sum_k\pi_k H(\Pi(|k\rangle\langle
k|)\|\Pi(\rho_*))=H(\rho_*)=C_{\mathrm{ea}}(\Pi,H,h).
$$
\end{example}\medskip

Proposition~\ref{degradable} is generalized as follows.\smallskip

\begin{property}\label{degradable-cv}
\emph{Let\/ $\Phi\colon\mathfrak{S}(\mathcal{H}_A)\to\mathfrak{S}(\mathcal{H}_B)$ be a degradable
channel,
$H$ a positive operator, $h>0$ and $\,h_*=(\dim\mathcal{H}_A)^{-1}\Tr H$. Then one of the following alternatives holds:
\begin{itemize}
\item
$\bar{C}(\Phi,H,h)<C_{\mathrm{ea}}(\Phi,H,h)$;
\item
$\Phi$ is a classical-quantum channel having the representation
\begin{equation}\label{c-q-rep++}
\Phi(\rho)=\sum_{k=1}^{\dim\mathcal{H}_A}\langle k|\rho|k\rangle\sigma_k,\quad
\rho\in\mathfrak{S}(\mathcal{H}_A),
\end{equation}
where  $\{\sigma_k\}$ is a collection of states in $\mathfrak{S}(\mathcal{H}_B)$ with mutually orthogonal supports and   $\{|k\rangle\}$
\renewcommand{\labelitemii}{\strut}
\begin{itemize}
\item\hskip-5pt
- is an orthonormal basis in $\mathcal{H}_A$, if $h\ge h_*$;
\item\hskip-5pt
- is the orthonormal basis of eigenvectors of the operator $H$, if $h<h_*$.
\end{itemize}
\end{itemize}
}\end{property}

\begin{proof}
Since $\chi^{}_\Phi(\rho)\le H(\rho)$ and $I(\rho,\Phi)\ge H(\rho)$ ($\Phi$ is a degradable channel),
the equality $\bar{C}(\Phi,H,h)=C_{\mathrm{ea}}(\Phi,H,h)$ may be valid only if
$$
\bar{C}(\Phi,H,h)=C_{\mathrm{ea}}(\Phi,H,h)=\max_{\Tr H\rho\le h}H(\rho).
$$

If $h\ge h_*$ then this maximum coincides with $\log\dim\mathcal{H}_A$, which means that the constraint has no effect and hence the second alternative in Proposition~\ref{degradable} holds.\smallskip

If $h<h_*$ then the above maximum is always attained at a full rank state and the generalized version of Theorem~\ref{e-b-ch},\,B) implies that $\Phi$ is a classical-quantum channel having representation~\eqref{c-q-rep++}. Similar to the proof of Proposition~\ref{degradable} one can show that
the states in the collection $\{\sigma_k\}$ have mutually orthogonal supports.

Show that the equality $\bar{C}(\Phi,H,h)=C_{\mathrm{ea}}(\Phi,H,h)$ may be valid in the case $h<h_*$ if and only if the operator $H$ is diagonizable in the basis $\{|k\rangle\}$ from representation
\eqref{c-q-rep++} of the channel $\Phi$. For the channel $\Pi(\rho)=\sum_k\langle
k|\rho|k\rangle|k\rangle\langle k|$ this assertion is proved in Example \ref{classical-costrained}. To prove it in general case it suffices to note that $\bar{C}(\Phi,H,h)=\bar{C}(\Pi,H,h)$ and $C_{\mathrm{ea}}(\Phi,H,h)=C_{\mathrm{ea}}(\Pi,H,h)$. These equalities follow from the chain rules for the capacities, since it is easy to construct channels $\Psi_1$ and $\Psi_2$ such that
$\Pi=\Psi_1\circ\Phi$ and $\Phi=\Psi_2\circ\Pi$.
\end{proof}\smallskip

The following proposition shows that coincidence of
$\bar{C}(\Phi,H,h)$ and\break  $C_{\mathrm{ea}}(\Phi,H,h)$ for any
constraint parameters $(H,h)$ is a very strong requirement.
\smallskip

\begin{property}\label{e-b-ch-c+}
\emph{If\/ $\Phi\colon\mathfrak{S}(\mathcal{H}_A)\to\mathfrak{S}(\mathcal{H}_B)$ is a quantum channel such that
$C_{\mathrm{ea}}(\Phi,H,h)=\bar{C}(\Phi,H,h)$ for any operator $H\ge0$ and $h>0$ then\/ $\Phi$ is a classical-quantum channel such that
\/ $\chi^{}_{\widehat{\Phi}}(\rho)=H(\rho)$ for all \/ $\rho\in\mathfrak{S}(\mathcal{H}_A)$. If the below Conjecture is true then  $\Phi$ is
the completely depolarizing channel.}
\end{property}

\begin{proof}
By  Lemma 1 in~\cite{H-Sh-1} an arbitrary full rank
state $\rho$ in $\mathfrak{S}(\mathcal{H}_A)$ can be made the average state of an
optimal ensemble for the channel $\Phi$ with constraint~\eqref{lc}
by appropriate choice of the operator $H$. Hence the condition of
the proposition and continuity arguments imply
$I(\rho,\Phi)=\chi^{}_\Phi(\rho)$ for any state $\rho$ in $\mathfrak{S}(\mathcal{H}_A)$.
By expression~\eqref{mi-rep} this means that $\chi^{}_{\widehat{\Phi}}(\rho)=H(\rho)$ for any state $\rho$
in $\mathfrak{S}(\mathcal{H}_A)$. By the generalized version of Theorem~\ref{e-b-ch},\,B) $\Phi$ is a classical-quantum channel.
\end{proof}

\textbf{Conjecture.} \emph{If \/
$\Phi\colon\mathfrak{S}(\mathcal{H}_A)\to\mathfrak{S}(\mathcal{H}_B)$ is a quantum
channel such that \/
$\chi^{}_\Phi(\rho)=H(\rho)$ for all \/ $\rho\in\mathfrak{S}(\mathcal{H}_A)$
then the channel \/ $\Phi$ coincides (up to unitary
equivalence) with the channel \/ $\rho\mapsto\rho\otimes\sigma$ for
some state $\sigma$.}

\section{The function $\Delta_\Phi(\rho)=I(\rho,
\Phi)-\chi^{}_\Phi(\rho)$ and its maximal value}

Central role in analysis of relations between
entanglement-assisted and unassisted classical capacities of a
quantum channel $\Phi$ is played by the function
$$
\Delta_\Phi(\rho)=I(\rho, \Phi)-\chi^{}_\Phi(\rho)
$$
introduced in Section~2, where it was mentioned that
$$
\Delta_\Phi(\rho)=H(\rho)-\chi^{}_{\widehat{\Phi}}(\rho)
=\min_{\substack{\sum_i\pi_i\rho_i=\rho\\ \rank\rho_i=1}}\,\sum_i\pi_i
\left[H(\rho_i\|\rho)-H(\widehat{\Phi}(\rho_i)\|\widehat{\Phi}(\rho))\right]
$$
and that the above minimum is achieved at an ensemble
$\{\pi_i,\rho_i\}$ of pure states if and only if this ensemble is
$\chi^{}_\Phi$-optimal in the sense of the following
definition.\smallskip

\begin{definition}\label{chi-opt}
An ensemble $\{\pi_i,\rho_i\}$ of pure states is called \emph{$\chi^{}_\Phi$-optimal} if the maximum in definition~\eqref{chi-fun} of the $\chi$-function of the channel $\Phi$ is
achieved at this ensemble.
\end{definition}\smallskip

Since $\hat{H}_\Phi\equiv\hat{H}_{\widehat{\Phi}}$, any $\chi^{}_\Phi$-optimal ensemble
is $\chi^{}_{\widehat{\Phi}}$-optimal and vice versa.\smallskip

The above formula for the function $\Delta_\Phi$ and monotonicity of the relative entropy imply the following observation.
\smallskip

\begin{lemma}\label{last-l}
\emph{If $\,\Phi$ is a degradable channel then  $\Delta_\Phi(\rho)\ge\Delta_{\widehat{\Phi}}(\rho)$ for all  $\,\rho$.}
\end{lemma}
\smallskip

In the following theorem properties of the function
$\Delta_\Phi$ are described.\smallskip

\begin{theorem}\label{delta-fun-p}
\emph{Let\/ $\Phi\colon\mathfrak{S}(\mathcal{H}_A)\to\mathfrak{S}(\mathcal{H}_B)$ be a quantum channel and\break
\/ $\widehat{\Phi}\colon\mathfrak{S}(\mathcal{H}_A) \to\mathfrak{S}(\mathcal{H}_E)$ its complementary
channel. \/ $\Delta_\Phi$ is a nonnegative continuous function on
the set \/ $\mathfrak{S}(\mathcal{H}_A)$
equal to zero on the subset
\/ $\mathop{\rm extr}\mathfrak{S}(\mathcal{H}_A)$ of pure states. It has the following
properties:}
\begin{enumerate}[\rm1)]
\item
\emph{if there exists a channel \/ $\Theta\colon\mathfrak{S}(\mathcal{H}_E)\to\mathfrak{S}(\mathcal{H}_A)$ such that
\begin{equation}\label{theta-2}
\Theta(\widehat{\Phi}(\rho_i))=\rho_i,\quad \forall i,
\end{equation}
for some ensemble \/$\{\pi_i,\rho_i\}$ of pure states with the average state \/$\rho$
    then \/ $\Delta_\Phi(\rho)=0$ and the  ensemble \/ $\{\pi_i,\rho_i\}$ is
\/ $\chi^{}_\Phi$-optimal;}
\item
\emph{if\/ $\Delta_\Phi(\rho)=0$ then
\begin{itemize}
\item
\eqref{theta-2} holds for any $\chi^{}_\Phi$-optimal ensemble
     $\{\pi_i,\rho_i\}$ with the average state  $\rho$, where $\Theta$
is a channel acting on a state $\sigma$ supported by the subspace $\supp\widehat{\Phi}(\rho)$ as follows:
$\Theta(\sigma)=A\widehat{\Phi}^*(B\sigma B)A$, $A=\rho^{1/2}$,
$B=\widehat{\Phi}(\rho)^{-1/2}$;
\item
$\Phi|_{\mathfrak{S}(\mathcal{H}_{\rho})}$ is a classical-quantum subchannel of the channel $\,\Phi$, where
$\mathcal{H}_{\rho}$ is the support of the state $\rho$;
\item
$\Delta_\Phi(\sum_i\lambda_i\rho_i)=0$ for any $\chi^{}_\Phi$-optimal ensemble $\{\pi_i,\rho_i\}$ with the average state $\rho$ and any probability distribution  $\{\lambda_i\}$.
\end{itemize}}
\item
\emph{the function $\Delta_\Phi$ is concave on the set \footnote{The function $\Delta_\Phi$ is not concave on $\mathfrak{S}(\mathcal{H}_A)$ in general, since otherwise we would obtain
$\Delta_\Phi(\rho)\le\Delta_\Phi(\rho_c)=0$ for any covariant
channel $\Phi$ such that
$C_{\mathrm{ea}}(\Phi)=\bar{C}(\Phi)$.} \/ $\Bigl\{\sum_i\lambda_i\rho_i\mid
\sum_i\lambda_i=1,\: \lambda_i\ge0\Bigr\}$ \/ for any $\chi^{}_\Phi$-optimal ensemble
$\{\pi_i,\rho_i\}$;}
\item
monotonicity: \emph{for an arbitrary  channel
\/ $\Psi\colon\mathfrak{S}(\mathcal{H}_B)\to\mathfrak{S}(\mathcal{H}_C)$ the following inequality holds}
$$
\Delta_{\Psi\circ\Phi}(\rho)\le\Delta_\Phi(\rho),\quad
\rho\in\mathfrak{S}(\mathcal{H}_A);
$$
\item
subadditivity for tensor product states: \emph{for an arbitrary
quantum channel \/ $\Psi\colon\mathfrak{S}(\mathcal{H}_C)\to\mathfrak{S}(\mathcal{H}_D)$ the following
inequality holds:
$$
\Delta_{\Phi\otimes\Psi}(\rho\otimes\sigma)\le\Delta_\Phi(\rho)+\Delta_{\Psi}(\sigma),
\quad\rho\in\mathfrak{S}(\mathcal{H}_A),\quad \sigma\in\mathfrak{S}(\mathcal{H}_C),
$$
which is satisfied with an equality if the strong additivity
of the Holevo capacity holds for the channels $\,\Phi$ and $\,\Psi$  (see~{\cite{H-Sh-1}}).}
\end{enumerate}
\end{theorem}

\begin{proof}
1) This property follows from monotonicity of the
relative entropy and the remark before Definition~\ref{chi-opt}.\smallskip

2) The first assertion follows from the Petz theorem~\cite[Theorem 3]{Petz} characterizing the case in which monotonicity of the
relative entropy holds with an equality.

The second assertion is derived from the first  one by using the arguments from
the proof of Theorem~\ref{e-b-ch},B\,).

The third assertion follows from the first one and property 1).\smallskip

3) Since $\hat{H}_\Phi\equiv\hat{H}_{\widehat{\Phi}}$, representation \eqref{chi-fun+} for the function
$\chi^{}_{\widehat{\Phi}}$ implies
$$
\Delta_\Phi(\rho)=\left[H(\rho)-H(\widehat{\Phi}(\rho))\right]+\hat{H}_\Phi(\rho).
$$

By the identity $H(\bar{\rho})-\sum_i\pi_iH(\rho_i)=\sum_i\pi_iH(\rho_i\|\bar{\rho})$,
where $\bar{\rho}=\sum_i\pi_i\rho_i$, concavity of the term in
the square brackets on the set $\mathfrak{S}(\mathcal{H}_A)$ follows from monotonicity of the relative entropy. So, to prove this
assertion it suffices to show that the function $\hat{H}_{\Phi}$ is
affine on the set
$\Bigl\{\sum_i\lambda_i\rho_i\mid\sum_i\lambda_i=1,\:
\lambda_i\ge0\Bigr\}$. This
can be done by noting that the function $\hat{H}_{\Phi}$ coincides
with the double Fenchel transform of the function $H\circ\Phi$ and
by using Proposition 1 in~\cite{A&B}.\smallskip

4) By using the Stinespring representation it is easy to show (see~\cite[the proof of Lemma 17]{R&Co}) that there exists a channel
$\Theta$ such that
$\widehat{\Phi}=\Theta\circ\widehat{\Psi\circ\Phi}$. Hence
the chain rule for the $\chi$-function implies
$$
\Delta_{\Psi\circ\Phi}(\rho)=H(\rho)-\chi_{\widehat{\Psi\circ\Phi}}(\rho)\le
H(\rho)-\chi^{}_{\widehat{\Phi}}(\rho)=\Delta_\Phi(\rho).
$$

5) Since $\widehat{\Phi\otimes\Psi}=\widehat{\Phi}\otimes\widehat{\Psi}$ (see
\cite{H-comp-ch}), this assertion follows from the obvious inequality $\chi_{\widehat{\Phi}\otimes\widehat{\Psi}}(\rho\otimes\sigma)
\ge\chi^{}_{\widehat{\Phi}}(\rho)+\chi_{\widehat{\Psi}}(\sigma)$, which is satisfied with an equality if the strong additivity of the
Holevo capacity holds for the channels $\Phi$ and $\Psi$~\cite{H-Sh-1}.
\end{proof}\smallskip

The following proposition shows the sense of the maximal value of
the function~$\Delta_\Phi$.\smallskip

\begin{property}\label{delta-fun-max}
\emph{Let \/ $\Phi\colon\mathfrak{S}(\mathcal{H}_A)\to\mathfrak{S}(\mathcal{H}_B)$ be a quantum channel. Then
\begin{equation}\label{max=sup}
\max_{\rho\in\mathfrak{S}(\mathcal{H}_A)}\Delta_\Phi(\rho)=\sup_{H,h}
\left[C_{\mathrm{ea}}(\Phi,H,h)-\bar{C}(\Phi,H,h)\right],
\end{equation}
where the supremum is over all pairs (positive operator $H\in\mathfrak{B}(\mathcal{H}_A)$, $h>0$).}
\end{property}

\begin{proof}
For given $H$ and $h$ let $\rho$ be a state in $\mathfrak{S}(\mathcal{H}_A)$ such that $\Tr H\rho\le h$ and
$C_{\mathrm{ea}}(\Phi,H,h)=I(\rho,\Phi)$. Since $\bar{C}(\Phi,H,h)\ge\chi^{}_\Phi(\rho)$, we have
$$
\Delta_\Phi(\rho)=I(\rho,\Phi)-\chi^{}_\Phi(\rho)\ge
C_{\mathrm{ea}}(\Phi,H,h)-\bar{C}(\Phi,H,h),
$$
This implies ``\,$\ge$\,'' in~\eqref{max=sup}.

Let $\varepsilon>0$ be arbitrary and $\rho_{\varepsilon}$ be a full
rank state in $\mathfrak{S}(\mathcal{H}_A)$ such that
$\Delta_\Phi(\rho_{\varepsilon})
\ge\max_{\rho\in\mathfrak{S}(\mathcal{H}_A)}\Delta_\Phi(\rho)-\varepsilon$. By Lemma 1 in~\cite{H-Sh-1}
there exists a pair $(H,h)$ such that $\Tr H\rho_{\varepsilon}\le h$ and
$\bar{C}(\Phi,H,h)=\chi^{}_\Phi(\rho_{\varepsilon})$. Since $C_{\mathrm{ea}}(\Phi,H,h)\ge
I(\rho_{\varepsilon},\Phi)$, we have
$$
C_{\mathrm{ea}}(\Phi,H,h)-\bar{C}(\Phi,H,h)\ge
I(\rho_{\varepsilon},\Phi)-\chi^{}_\Phi(\rho_{\varepsilon})=
\Delta_\Phi(\rho_{\varepsilon})
\ge\max_{\rho\in\mathfrak{S}(\mathcal{H}_A)}\Delta_\Phi(\rho)-\varepsilon,
$$
which implies ``\,$\le$\,'' in~\eqref{max=sup}.
\end{proof}

It is easy to see that $\max_{\rho\in\mathfrak{S}(\mathcal{H}_A)}\Delta_\Phi(\rho)
\in[0,\log\dim\mathcal{H}_A]$. If $\Delta_\Phi(\rho)\equiv0$ then the condition of Proposition~\ref{e-b-ch-c+} holds. If
$\max_{\rho\in\mathfrak{S}(\mathcal{H}_A)}\Delta_\Phi(\rho)
=\log\dim\mathcal{H}_A$ then $\Phi$
is unitary equivalent to the channel $\rho\mapsto\rho\otimes\sigma$, where $\sigma$ is a given state.
Indeed, this implies $\chi^{}_{\widehat{\Phi}}(\rho_c)=0$, where $\rho_c$ is the chaotic state in $\mathfrak{S}(\mathcal{H}_A)$, and hence
$\chi^{}_{\widehat{\Phi}}(\rho)\equiv0$ by concavity and nonnegativity of the $\chi$-function, which means that~$\widehat{\Phi}$~is a completely depolarizing channel.\medskip

\begin{remark}\label{delta-fun-max-r}
Subadditivity of the function
$\Delta_\Phi$ (property 5 in Theorem~\ref{delta-fun-p}) implies
existence of the regularization
$\Delta^*_\Phi(\rho)=\lim_{n\to+\infty}n^{-1}\Delta_{\Phi^{\otimes
n}}(\rho^{\otimes n})$. By repeating the arguments from the proof of
Proposition \ref{delta-fun-max} and by using subadditivity of the
quantum mutual information it is easy to show that
$$
\max_{\rho\in\mathfrak{S}(\mathcal{H}_A)}\Delta^*_\Phi(\rho)\ge\sup_{H,h}
\left[C_{\mathrm{ea}}(\Phi,H,h)-C(\Phi,H,h)\right].
$$
The equality in this inequality is obvious if the strong additivity of the Holevo
capacity holds for the channel $\Phi$  (see \cite{H-Sh-1}), but it seems to be not valid in general.
\end{remark}\smallskip

Let $\Phi\colon\mathfrak{S}(\mathcal{H}_A)\to\mathfrak{S}(\mathcal{H}_B)$
and $\Psi\colon\mathfrak{S}(\mathcal{H}_B)\to\mathfrak{S}(\mathcal{H}_C)$ be quantum channels. Monotonicity
of the function $\Delta_\Phi$ (property 4 in Theorem
~\ref{delta-fun-p}) shows that the inequality
$$
C_{\mathrm{ea}}(\Psi\circ\Phi,H,h)-\bar{C}(\Psi\circ\Phi,H,h)\le
C_{\mathrm{ea}}(\Phi,H,h)-\bar{C}(\Phi,H,h)
$$
is valid if the functions $\rho\mapsto I(\rho,\Psi\circ\Phi)$ and $\rho\mapsto\chi^{}_\Phi(\rho)$ have common maximum point under the
condition $\Tr H\rho\le
h$ (this holds for the unconstrained
channels $\Phi$ and $\Psi$ satisfying the covariance condition~\eqref{cov} with $\mathcal{H}_A=\mathcal{H}_B$ and $V_g=W_g$).

In general validity of the above  inequality is an interesting open
question, but monotonicity of the function $\Delta_\Phi$ and
Proposition~\ref{delta-fun-max} imply the following
observation.\smallskip

\begin{corollary}\label{delta-fun-max-c}
\emph{Let \/ $\Phi\colon\mathfrak{S}(\mathcal{H}_A)\to\mathfrak{S}(\mathcal{H}_B)$ and \/
$\Psi\colon\mathfrak{S}(\mathcal{H}_B)\to\mathfrak{S}(\mathcal{H}_C)$ be arbitrary quantum channels.
Then}
$$
\sup_{H,h}
\left[C_{\mathrm{ea}}(\Psi\circ\Phi,H,h)-\bar{C}(\Psi\circ\Phi,H,h)\right]\le\sup_{H,h}
\left[C_{\mathrm{ea}}(\Phi,H,h)-\bar{C}(\Phi,H,h)\right].
$$
\end{corollary}

By introducing  the parameter
$$
D(\Phi)=\sup_{H,h} \left[C_{\mathrm{ea}}(\Phi,H,h)-\bar{C}(\Phi,H,h)\right]
$$
of the channel $\Phi\colon\mathfrak{S}(\mathcal{H}_A)\to\mathfrak{S}(\mathcal{H}_B)$ the above
observations can be reformulated as follows:
\smallskip
\begin{itemize}
\item
$D(\Phi)=\max_{\rho\in\mathfrak{S}(\mathcal{H}_A)}\Delta_\Phi(\rho)$;
\item
$D(\Psi\circ\Phi)\le D(\Phi)$ for any channel
$\Psi\colon\mathfrak{S}(\mathcal{H}_B)\to\mathfrak{S}(\mathcal{H}_C)$;
\item
$D(\Phi)\in[0,\log\dim\mathcal{H}_A]$;
\item
$D(\Phi)=\log\dim\mathcal{H}_A$ if and only if the channel $\Phi$ is unitary equivalent to the noiseless channel $\rho\mapsto\rho\otimes\sigma$, where $\sigma$ is a given state;
\item
$D(\Phi)=0$ if $\Phi$ is a completely depolarizing channel ("if and only if" provided the Conjecture at the end of Section~3 is true).
\end{itemize}

\smallskip
The above properties  show that the parameter $D(\Phi)$ can be
considered as one of characteristics of the channel $\Phi$
describing its "level of noise". Unfortunately, this parameter seems
not to be easily calculated for nontrivial examples of quantum
channels.\medskip

Generalizations of the results obtained in this paper to infinite dimensional constrained channels are presented in the second part of~\cite{Sh}.

\bigskip
I am grateful to A.S.Holevo and to the participants of his seminar
"Quantum probability, statistic, information" (the Steklov
Mathematical Institute) for useful discussion.

\smallskip

The work is supported in part by the
Scientific Program ``Mathematical Control Theory and Dynamic Systems'' of the Russian
Academy of Sciences and the Russian Foundation for Basic Research, projects
10-01-00139-a and~12-01-00319-a.

\end{document}